\documentclass[letterpaper, 10 pt, conference]{ieeeconf}

\IEEEoverridecommandlockouts                            
\usepackage{amsmath,amssymb,booktabs,colortbl,cases}
\usepackage{relsize}
\usepackage{dsfont}
\usepackage{cuted}
\usepackage{url,subcaption}
\usepackage[latin1]{inputenc}
\usepackage{eufrak}
\usepackage{epsfig}
\usepackage{color}
\usepackage[noadjust]{cite}
\usepackage{mathrsfs}
\usepackage{tikz,lipsum}
\usepackage{enumerate}
\usepackage{psfrag}
\usepackage{adjustbox,graphics} 
\usetikzlibrary{arrows,shapes.geometric,positioning}
\usepackage{epsfig} 
\definecolor{my_GREEN}{rgb}{0,0.5,0}
\definecolor{LightCyan}{rgb}{0.88,1,1}
\definecolor{my_Gray}{rgb}{0.85,0.85,0.85}
\definecolor{myCOLOR1}{RGB}{0,0,255}
\newcommand{\editKG}[1]{{\textcolor{black}{#1}}}
\allowdisplaybreaks
\newtheorem{theorem}{Theorem}[section]

\newtheorem{proposition}[theorem]{Proposition}

\newtheorem{remark}[theorem]{Remark}
\makeatletter

\newcommand{\Rmnum}[1]{\expandafter\@slowromancap\romannumeral #1@}

\makeatother
\title{\LARGE \bf
Beacon-referenced Mutual Pursuit in Three Dimensions
}
\author{Kevin S. Galloway$^{\dagger}$ and Biswadip Dey$^{\ddagger}$
\thanks{$^{\dagger}$Kevin S. Galloway is with the Electrical and Computer Engineering Department, United States Naval Academy, Annapolis, MD 21402 USA. {\tt\small kgallowa@usna.edu}}
\thanks{$^{\ddagger}$Biswadip Dey is with the Department of Mechanical and Aerospace Engineering, Princeton University, Princeton, NJ 08544, USA. {\tt\small biswadip@princeton.edu}}
\thanks{*The second author's research was supported in part by the Office of Naval Research under ONR grant N00014-14-1-0635.}
}
\begin{document}
\newcommand{\RR}{\mathds{R}}
\newcommand{\rUnit}{\frac{\bf r}{\left|{\bf r}\right|}}
\newcommand{\zBarUnit}{\frac{\bar{\bf z}}{\left|\bar{\bf z}\right|}}
\newcommand{\riUnit}{\frac{{\bf r}_{i,i+1}}{\left|{\bf r}_{i,i+1}\right|}}
\newcommand{\riBUnit}{\frac{{\bf r}_{ib}}{\left|{\bf r}_{ib}\right|}}
\newcommand{\riBPlusUnit}{\frac{{\bf r}_{i+1,b}}{\left|{\bf r}_{i+1,b}\right|}}
\newcommand{\riPlusUnit}{\frac{{\bf r}_{i+1,i+2}}{\left|{\bf r}_{i+1,i+2}\right|}}
\newcommand{\riMinusUnit}{\frac{{\bf r}_{i-1,i}}{\left|{\bf r}_{i-1,i}\right|}}
\newcommand{\Bl}{\boldsymbol \ell}
\newcommand{\aP}{\bar{a}_{+}}
\newcommand{\aM}{\bar{a}_{-}}
%
%
\maketitle
\thispagestyle{empty}
\pagestyle{empty}
%
%
\begin{abstract}
Motivated by station-keeping applications in various unmanned settings, this paper introduces a steering control law for a pair of agents operating in the vicinity of a fixed beacon in a three-dimensional environment. This feedback law is a modification of the previously studied three-dimensional constant bearing (CB) pursuit law, in the sense that it incorporates an additional term to allocate attention to the beacon. We investigate the behavior of the closed-loop dynamics for a two-agent \emph{mutual pursuit} system in which each agent employs the beacon-referenced CB pursuit law with regards to the other agent and a stationary beacon. Under certain assumptions on the associated control parameters, we demonstrate that this problem admits \emph{circling equilibria} wherein the agents move on circular orbits with a common radius, in planes perpendicular to a common axis passing through the beacon. As the common radius and distances from the beacon are determined by choice of parameters in the feedback law, this approach provides a means to engineer desired formations in a three-dimensional setting.
\end{abstract}
%
%

%
%
%
\section{Introduction}
\label{sec:Intro}
%
As pursuit and collective motion play significant roles in various contexts of robotics and engineering, it seems appealing to seek inspiration from nature, which abounds with many such examples. \editKG{Among the} various possible ways to pursue and intercept a moving target, evidence of constant bearing (CB) pursuit strategy can be \editKG{observed} in a variety of animal species (e.g. dogs \cite{Shaffer_Dog_2004}, flies \cite{Collett1975,Osorio281}, humans \cite{CHARDENON200213}, raptors \cite{Tucker3745}). The CB pursuit strategy dictates that an agent should move towards its target in such a way that the angle between the baseline (alternatively known as the line-of-sight) connecting the two individuals and its own velocity remains constant. By prescribing a fixed offset between the baseline and the pursuer's velocity, this strategy provides further generalization of the classical pursuit strategy (wherein the pursuer always moves directly towards the current location of its target).

The work in \cite{Pavone_ASME_commonOffset} exploited this pursuit strategy as a building block for designing formations in an engineered setting, and demonstrated that, by applying a homogeneous\footnote{\textit{Homogeneous} in the sense that each individual attempts to maintain the same angular offset between its velocity and the baseline towards neighbor.} CB pursuit strategy in a cyclic manner, a collective of agents eventually converge to a common rendezvous point, a circular formation or a logarithmic spiral pattern. Another line of work \cite{Kevin_2011_CDC, Galloway_PRS_13,Galloway_PRS_16} explored a more general setting wherein individual agents pursue each other using a heterogeneous CB pursuit strategy, and has shown existence and stability of a richer class of behaviors (circular motion, rectilinear motion, shape preserving spirals and periodic orbits). While a majority of the prior research has considered only planar settings to investigate CB pursuit strategy in a cyclic interaction, the work in \cite{CP_3D_Kevin_2010,Galloway_PRS_16} considers the three-dimensional setting as well.

While this line of research has demonstrated existence of circling equilibria in which agents moved on a common circular trajectory, both the location of the circumcenter of the formation (with respect to some inertial frame) and its size were determined by initial conditions rather than control parameters. To overcome this aspect and to broaden its scope from a design perspective, we introduced a modified version of the CB control law in our earlier work \cite{KSG_BD_ACC_2015, KSG_BD_ACC_2016, KSG_BD_arXiv_17}. In this new setting, the pursuer pays attention to a beacon (which can represent an attractive food source in a biological setting, or some target of interest for an unmanned vehicle), in addition to its neighbor. Another line of work \cite{Mallik2015ConsensusApplications, Daingade2015AImplementation} has also investigated this beacon-referenced (or target-centric) cyclic pursuit framework, albeit their work uses a different control formulation. 

In the current work, we extend this beacon-referenced approach to the three dimensional-setting. We first introduce a beacon-referenced version of the CB pursuit law in three dimensions, and then the mutual pursuit scenario in which two agents apply this pursuit law to one another as well as a stationary beacon. Earlier work \cite{MISCHIATI20114483, MISCHIATI2012894, UH_BD_ICRA_15} has demonstrated that mutual pursuit can lead to a variety of interesting motion patterns, while providing better tractability from an nonlinear analysis perspective, and it can be viewed as a building block towards understanding the more general cyclic pursuit framework.

This paper is organized as follows. We begin by describing the self-steering particle model for agents moving in the three dimensions. Then, in the later part of Section~\ref{sec:Prob_SetUp}, we introduce the beacon-referenced constant bearing pursuit law. As the underlying self-steering particle model has a 1-dimensional rotational invariance, we can describe the evolution of the mutual pursuit system by considering a reduced system evolving on $\mathds{R}^3 \times \mathds{R}^3 \times \mathcal{S}^2 \times \mathcal{S}^2$. In the rest of Section~\ref{sec:Closed-Loop-Dyna}, we define the effective shape space, identify the geometric constraints associated with these variables, and derive the closed loop dynamics after making some simplifying assumptions about the control parameters. In Section~\ref{sec:Rel_Equilibrium}, we analyze the closed loop shape dynamics, and explore existence conditions and characterization of the associated relative equilibria. Finally we conclude in Section~\ref{sec:Conclusion}.
%
%
%

%
%
%
\section{Modeling Mutual Interactions}
\label{sec:Prob_SetUp}
\subsection{Generative Model: Agents as Self-Steering Particles}
%
Similar to earlier works (\cite{Justh_PSK_3Dformation, CP_3D_Kevin_2010}), we treat the agents as unit-mass self-steering particles moving along twice-differentiable paths in a three-dimensional environment. This allows us to describe the motion of an agent in terms of its natural Frenet frame \cite{Nat_Frenet_Bishop}, defined by its position $\mathbf{r}_i$ (with respect to an inertial reference frame) and an orthonormal triad of vectors $[\mathbf{x}_i,\mathbf{y}_i,\mathbf{z}_i]$. Then, by constraining the agents to move at equal and nonvanishing speed, we can assume without loss of generality that the agents are moving at unit speed, and express the dynamics of a pair of agents as
\begin{equation}
\begin{aligned}
\dot{\mathbf{r}}_i &= \mathbf{x}_i \\
\dot{\mathbf{x}}_i &= u_i \mathbf{y}_i + v_i \mathbf{z}_i \\
\dot{\mathbf{y}}_i &= - u_i \mathbf{x}_i \\
\dot{\mathbf{z}}_i &= - u_i \mathbf{x}_i, 
\end{aligned} 
\label{Explicit_MODEL}
\end{equation}
for $i=1,2$. Here, $u_i$ and $v_i$ are the natural curvatures viewed as gyroscopic steering controls. Moreover, we assume that the \textit{beacon} is located at position $\mathbf{r}_b \in \mathds{R}^3$. Then it directly follows that $\mathcal{M}_c = SE(3) \times SE(3) \times \mathds{R}^3$ defines the underlying \textit{configuration space} of dimension 15. 

However, as we are only interested in the agents' motion relative to each other and to the beacon, we can formulate a reduction to the 9-dimensional \textit{shape space}, defined as $\mathcal{M}_s = \mathcal{M}_c / SE(3)$. Similar to the scalar shape variables employed in the planar case \cite{KSG_BD_ACC_2015}, we can define the following set of scalar variables (for $i=1,2$)
\begin{equation}
\begin{aligned}
&
\bar{x}_{i} \triangleq \mathbf{x}_i \cdot \frac{\mathbf{r}_{i,i+1}}{|\mathbf{r}_{i,i+1}|}, 
&&
\bar{y}_{i} \triangleq \mathbf{y}_i \cdot \frac{\mathbf{r}_{i,i+1}}{|\mathbf{r}_{i,i+1}|},
 &&
\bar{z}_{i} \triangleq \mathbf{z}_i \cdot \frac{\mathbf{r}_{i,i+1}}{|\mathbf{r}_{i,i+1}|},
\\
&
\bar{x}_{ib} \triangleq \mathbf{x}_i \cdot \frac{\mathbf{r}_{ib}}{|\mathbf{r}_{ib}|}, 
&&
\bar{y}_{ib} \triangleq \mathbf{y}_i \cdot \frac{\mathbf{r}_{ib}}{|\mathbf{r}_{ib}|}, 
&&
\bar{z}_{ib} \triangleq \mathbf{z}_i \cdot \frac{\mathbf{r}_{ib}}{|\mathbf{r}_{ib}|}, 
\\
& 
\rho_{ib} \triangleq |\mathbf{r}_{ib}|,
&&
\rho \triangleq |\mathbf{r}_{12}|, 
&&
\tilde{x} \triangleq \mathbf{x}_1 \cdot \mathbf{x}_2
\end{aligned}
\end{equation}
to parametrize the shape space $\mathcal{M}_s$. (See Fig. \ref{Scalar_Shapes}.) Here, $\mathbf{r}_{ij} = \mathbf{r}_i - \mathbf{r}_j$, $i,j \in \{1,2\}$ represents the position of agent $i$ relative to agent $j$, $\mathbf{r}_{1b}$ and $\mathbf{r}_{2b}$ represent the positions of agent 1 and 2 relative to the beacon located at $\mathbf{r}_b$, respectively, and addition in the index variables should be interpreted modulo $2$. (This convention will be employed throughout this work.) Clearly, these variables overparameterize the underlying shape space. However, this overparameterization can be taken into account by considering the appropriate constraints (\emph{e.g.} $\bar{x}_i^2 + \bar{y}_i^2 + \bar{z}_i^2 = 1$, $i = 1,2$).

In what follows, we will prescribe that the agents should not be collocated with each other or with the beacon, i.e. we assume $\rho>0$, $\rho_{1b}>0$, and $\rho_{2b}>0$. These assumptions are made to keep the pursuit laws well-defined, but are not necessarily enforced by the closed-loop system dynamics. 
\begin{figure}[t!]
\begin{center}
  \includegraphics[height=2in]{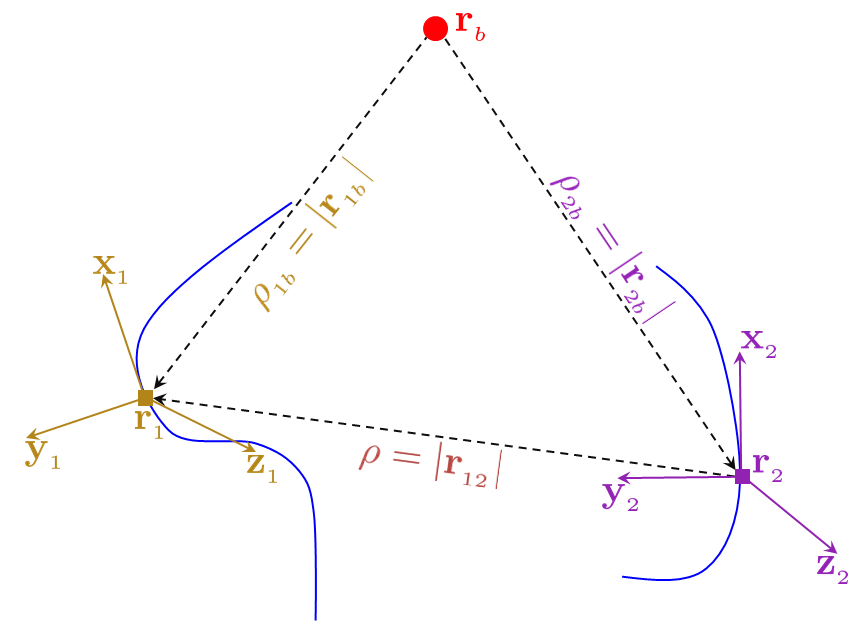}
  \caption{\small{Agent trajectories, together with the corresponding natural Frenet frames, for a beacon-referenced mutual pursuit system in a three-dimensional setting.}}
  \vspace{-1.5em}
  \label{Scalar_Shapes}
\end{center}
\end{figure}
%
%
%

%
%
%
\subsection{Beacon-referenced Constant Bearing Pursuit in Three Dimensions}
%
\editKG{Previous work in \cite{CP_3D_Kevin_2010,Galloway_PRS_16} introduced and analyzed a feedback pursuit law for executing the constant bearing (CB) strategy in three dimensions. In what follows, we propose a modified version of the CB that includes an additional term referenced to the bearing toward the beacon.} Similar to our previous work for the planar setting \cite{KSG_BD_ACC_2015}, we construct this feedback law as a convex combination of two fundamental building blocks, expressed as
\begin{equation}
\begin{aligned}
u_i &= (1 - \lambda)u_i^{CB} + \lambda u_i^B
\\
v_i &= (1 - \lambda)v_i^{CB} + \lambda v_i^B
\end{aligned}
\label{Steering_Feedback_top_level}
\end{equation}
for $i = 1,2$, where $\lambda \in [0,1]$ maintains a balance between the influence of the beacon and that of the neighboring agent. In this feedback law \eqref{Steering_Feedback_top_level}, $u_i^{CB}$, $v_i^{CB}$ are governed by the original CB pursuit law \cite{CP_3D_Kevin_2010}, and $u_i^B$, $v_i^B$ represent the deviation from a desired bearing toward the beacon, \editKG{as described in detail below.}

In particular, by letting $\mu_i > 0$ denote a positive control gain, we choose
\begin{subequations}
\begin{align}
u_i^{CB} &= -\mu_i (\bar{x}_i - a_i)\bar{y}_i 
\nonumber \\
& \qquad 
- \frac{1}{\left|{\bf r}_{i,i+1}\right|}\left[{\bf z}_i \cdot \left(\dot{\bf r}_{i,i+1} \times \riUnit \right) \right]
\label{u_CB_i} \\
v_i^{CB} &= -\mu_i (\bar{x}_i - a_i)\bar{z}_i 
\nonumber \\
& \qquad 
+ \frac{1}{\left|{\bf r}_{i,i+1}\right|}\left[{\bf y}_i \cdot \left(\dot{\bf r}_{i,i+1} \times \riUnit \right) \right],
\label{v_CB_i}
\end{align}
\end{subequations}
where the parameter $a_i \in [-1,1]$ represents the desired offset between the heading of agent $i$ and its bearing toward agent $(i+1)$. We choose the beacon tracking component as
\begin{subequations}
\begin{align}
u_i^{B} &= -\mu_i^b (\bar{x}_{ib} - a_{ib})\bar{y}_{ib}
\label{u_Beacon_i} \\
v_i^{B} &= -\mu_i^b (\bar{x}_{ib} - a_{ib})\bar{z}_{ib},
\label{v_Beacon_i}
\end{align}
\end{subequations}
where $\mu_i^b > 0$ is the corresponding control gain and the parameter $a_{ib} \in [-1,1]$ represents the desired offset between the heading of agent $i$ and its bearing toward the beacon. In general, the neighbor- tracking goal may conflict with the beacon-referencing goal, \editKG{i.e. there are no guarantees that both goals can be attained.} \editKG{Also, for $\lambda = 0$, \eqref{Steering_Feedback_top_level} simplies to the already analyzed CB pursuit law from \cite{CP_3D_Kevin_2010}, and for $\lambda = 1$ the system devolves to simple beacon-tracking by multiple independent agents. Therefore we will assume $\lambda \in (0,1)$ for the duration of this work.}
%
%
%

%
%
%
\section{Closed Loop Shape Dynamics}
\label{sec:Closed-Loop-Dyna}
%
In \cite{Justh_PSK_3Dformation}, the authors have demonstrated the importance of considering a reduced system evolving on 
$\mathds{R}^3 \times \mathcal{S}^2 \times \mathcal{S}^2$ for analyzing certain types of two-agent systems with their dynamics defined on $SE(3) \times SE(3)$. Before delving into further analysis, we investigate similar aspects for the system under consideration, and show existence of a corresponding \textit{reduced space}. We begin by computing
\begin{align}
\dot{\mathbf{x}}_i 
&=
(1 - \lambda)\Big[ u_i^{CB}\mathbf{y}_i + v_i^{CB}\mathbf{z}_i \Big] + 
\lambda \Big[ u_i^B\mathbf{y}_i + v_i^B\mathbf{z}_i\Big]
\nonumber \\
&=
-(1 - \lambda)\mu_i (\bar{x}_i - a_i)\Big[ \bar{y}_i\mathbf{y}_i + \bar{z}_i\mathbf{z}_i \Big] 
\nonumber \\
& \qquad 
-\lambda\mu_i^b (\bar{x}_{ib} - a_{ib}) \Big[ \bar{y}_{ib}\mathbf{y}_i + \bar{z}_{ib}\mathbf{z}_i\Big]
\nonumber \\
& \qquad
- \frac{(1 - \lambda)}{\left|{\bf r}_{i,i+1}\right|}
\left[\left({\bf z}_i \cdot \left(\dot{\bf r}_{i,i+1} \times \riUnit \right) \right)\mathbf{y}_i \right.
\nonumber \\
& \qquad \qquad 
- \left.\left({\bf y}_i \cdot \left(\dot{\bf r}_{i,i+1} \times \riUnit \right) \right)\mathbf{z}_i\right]
\label{lat_Accln}
\end{align}
for $i \in \{1,2\}$. Then by using the \emph{BAC-CAB} identity of vector algebra, we can express \eqref{lat_Accln} as
\begin{align}
\dot{\mathbf{x}}_i 
&=
-(1 - \lambda)\mu_i (\bar{x}_i - a_i) \left[ \frac{{\bf r}_{i,i+1}}{\left|{\bf r}_{i,i+1}\right|} - \bar{x}_i\mathbf{x}_i \right] 
\nonumber \\
& \qquad 
-\lambda\mu_i^b (\bar{x}_{ib} - a_{ib}) \left[ \frac{{\bf r}_{ib}}{\left|{\bf r}_{ib}\right|} - \bar{x}_{ib}\mathbf{x}_i \right]
\nonumber \\
& \qquad
+ \frac{(1 - \lambda)}{\left|{\bf r}_{i,i+1}\right|} \left[ {\bf x}_i \times \left(\dot{\bf r}_{i,i+1} \times \riUnit \right) \right].
\label{lat_Accln_Self_Contained}
\end{align}
\begin{figure}[t!]
\begin{center}
  \includegraphics[width=0.45\textwidth]{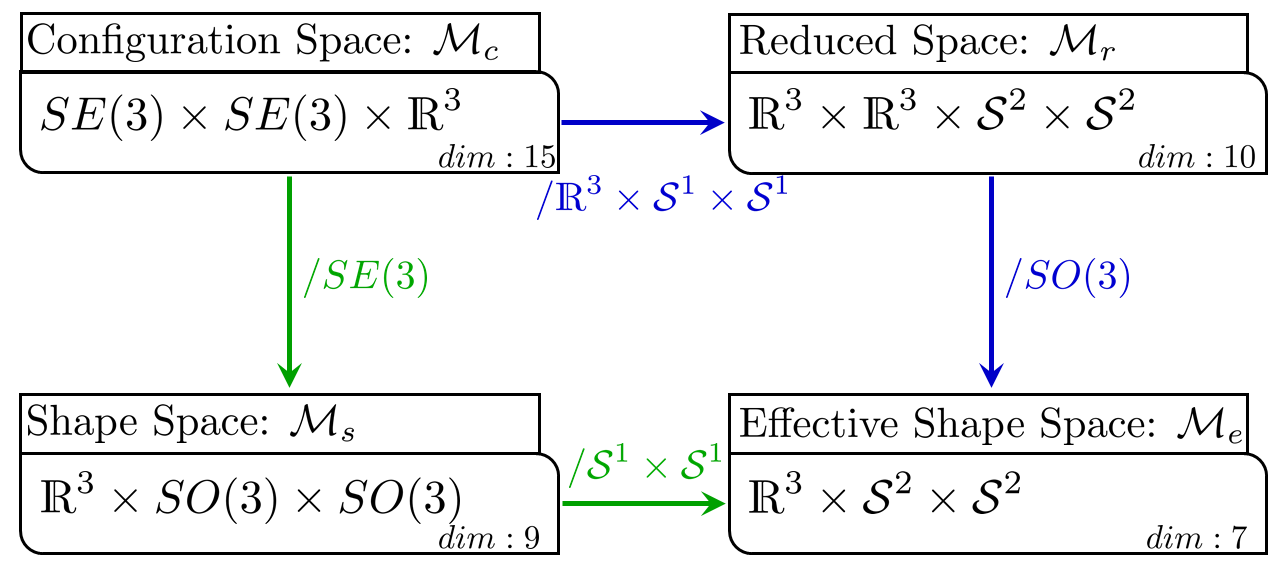}
  \caption{\small{Illustration of the reduction from the 15-dimensional configuration space $\mathcal{M}_c$ to a 7-dimensional effective shape space $\mathcal{M}_e$. The blue pathway involves an intermediate reduced space $\mathcal{M}_r$, while the green pathway passes through the shape space $\mathcal{M}_s$.}}
  \vspace{-1.5em}
  \label{Fig:Reduction_of_Spaces}
\end{center}
\end{figure}

As $\dot{\mathbf{r}}_{ib} = \mathbf{x}_i$ and $\mathbf{r}_{i,i+1}$ can also be expressed as $({\bf r}_{i,b} - {\bf r}_{i+1,b})$, it directly follows from \eqref{lat_Accln_Self_Contained} that the evolution of $(\mathbf{r}_{1b}, \mathbf{r}_{2b}, \mathbf{x}_1, \mathbf{x}_2)$ is governed by a self-contained dynamics on the \emph{reduced space} $\mathcal{M}_r = \mathds{R}^3 \times \mathds{R}^3 \times \mathcal{S}^2 \times \mathcal{S}^2$ of dimension 10, \editKG{as illustrated in Figure \ref{Fig:Reduction_of_Spaces}}. Then, after solving the evolution of this reduced dynamics, one can reconstruct the evolution of the complete frame $[\mathbf{x}_i, \mathbf{y}_i, \mathbf{z}_i]$ by using the rule of quadrature. With this observation, we focus on the reduced dynamics on $\mathcal{M}_r$, instead of the full dynamics defined on $\mathcal{M}_c$.

Furthermore, the reduced dynamics on $\mathcal{M}_r$ is invariant to any rotation with respect to an inertial reference. This allows us to carry out further reduction, and focus our attention to a reduced system defined on the 7-dimensional \emph{effective shape space} $\mathcal{M}_e$. As we will see in the later analysis, the following set of scalar variables provide an efficient parametrization of this effective shape space:
\begin{equation}
\begin{aligned}
&
\bar{x}_{1} = \mathbf{x}_1 \cdot \frac{\mathbf{r}_{12}}{|\mathbf{r}_{12}|}, 
&&
\bar{x}_{2} = \mathbf{x}_2 \cdot \frac{\mathbf{r}_{21}}{|\mathbf{r}_{21}|},
&&
\tilde{x} = \mathbf{x}_1 \cdot \mathbf{x}_2
\\
&
\bar{x}_{1b} = \mathbf{x}_1 \cdot \frac{\mathbf{r}_{1b}}{|\mathbf{r}_{1b}|}, 
&&
\bar{x}_{2b} = \mathbf{x}_2 \cdot \frac{\mathbf{r}_{2b}}{|\mathbf{r}_{2b}|}
&&
\\
& 
\rho_{1b} = |\mathbf{r}_{1b}|,
&&
\rho_{2b} = |\mathbf{r}_{2b}|,
&&
\rho = |\mathbf{r}_{12}|.
\end{aligned}
\label{scalar_Variables_EffShape}
\end{equation}
As we will see in the following subsection, these scalar variables \eqref{scalar_Variables_EffShape} are subject to appropriate constraints of codimension 1.
%
%

%
%
\subsection{Constraints on the Effective Shape Space Variables}
\label{sec:Constraints}
%
If the vectors $\mathbf{r}_{1b}$, $\mathbf{r}_{2b}$ and $\mathbf{r}_{12}$ are collinear, in addition to lying on the same plane (which directly follows from their definition), either of the following constraints shall hold true:
\begin{subequations}
\begin{align}
\rho_{1b} + \rho_{2b} &= \rho
\label{Constrain_1a}
\\
\textrm{or,} \qquad |\rho_{1b} - \rho_{2b}| &= \rho.
\label{Constrain_1b}
\end{align}
\end{subequations}

However, even if they are not collinear, we can still exploit the fact that $\mathbf{r}_{1b} - \mathbf{r}_{2b} = \mathbf{r}_{12}$, and obtain the relationships 
\begin{subequations}
\begin{align}
\rho_{1b} \bar{x}_{1b} - \rho_{2b} \left( \mathbf{x}_1 \cdot \frac{\mathbf{r}_{2b}}{|\mathbf{r}_{2b}|} \right)
&= 
\rho \bar{x}_1
\label{Constraints_2_1a}
\\
\textrm{and,} \quad
\rho_{1b} \left( \mathbf{x}_2 \cdot \frac{\mathbf{r}_{1b}}{|\mathbf{r}_{1b}|} \right) - \rho_{2b} \bar{x}_{2b} 
&= 
- \rho \bar{x}_2
\label{Constraints_2_1b}
\end{align} 
\end{subequations}
by taking their projections on the normalized velocities $\mathbf{x}_1$ and $\mathbf{x}_2$, respectively. As the dot-product of two unit vectors lies in the interval $[-1,1]$, \eqref{Constraints_2_1a}-\eqref{Constraints_2_1b} lead to the following inequality constraints:
\begin{subequations}
\begin{align}
&
- \rho_{2b} \leq \rho_{1b} \bar{x}_{1b} - \rho \bar{x}_1 \leq \rho_{2b}
\label{Constraints_2a}
\\
\textrm{and,} \quad
&
- \rho_{1b} \leq \rho_{2b} \bar{x}_{2b} - \rho \bar{x}_2 \leq \rho_{1b}.
\label{Constraints_2b}
\end{align} 
\end{subequations}
\editKG{Also, the Law of Cosines requires that
\begin{align}
\rho_{1b}-\rho_{2b} \leq \rho \leq \rho_{1b}+\rho_{2b},
\label{Constraints_2c}
\end{align}
with strict inequality if the agents are not collinear.}

In addition to these inequality constraints, we can also demonstrate that the underlying geometry leads to an additional constraint which poses restriction on the possible values of $\tilde{x}$ for some fixed values of $\bar{x}_{1}$, $\bar{x}_{2}$, $\bar{x}_{1b}$, $\bar{x}_{2b}$, $\rho$, $\rho_{1b}$ and $\rho_{2b}$, i.e. for the rest of the shape variables. As $\mathbf{r}_{1b}$, $\mathbf{r}_{2b}$ and $\mathbf{r}_{12}$ constitute a triangle, these three vectors lie on a plane. It readily follows that for a fixed value of $\bar{x}_{i}$, the normalized velocity vector $\mathbf{x}_{i}$ lies on a particular circle around $\mathbf{r}_{12}$ (or $\mathbf{r}_{21}$) which itself lies on the surface of a unit sphere. This circle is marked as $\mathcal{C}_i$ in Figure~\ref{Fig:Constraints}. In a similar way, a fixed value of $\bar{x}_{ib}$ forces $\mathbf{x}_{i}$ to lies on a particular circle around $\mathbf{r}_{ib}$ which itself lies on the surface of a unit sphere (shown as $\mathcal{C}_{ib}$ in Figure~\ref{Fig:Constraints}). Clearly, these two circles $\mathcal{C}_i$ and  $\mathcal{C}_{ib}$ intersect \editKG{(at most) at} two points $P_i'$ and $P_i''$. Furthermore, it can be shown that $P_i'$ and $P_i''$ are reflections of each other with respect to the plane containing $\mathbf{r}_{1b}$, $\mathbf{r}_{2b}$ and $\mathbf{r}_{12}$. As a consequence, $\tilde{x} = \mathbf{x}_1 \cdot \mathbf{x}_2$ can assume one out of only two possible values. In what follows, we will see that this constraint can be exploited in the analysis of the closed loop dynamics.
\begin{figure}[b!]
\begin{center}
  \includegraphics[width=0.3\textwidth]{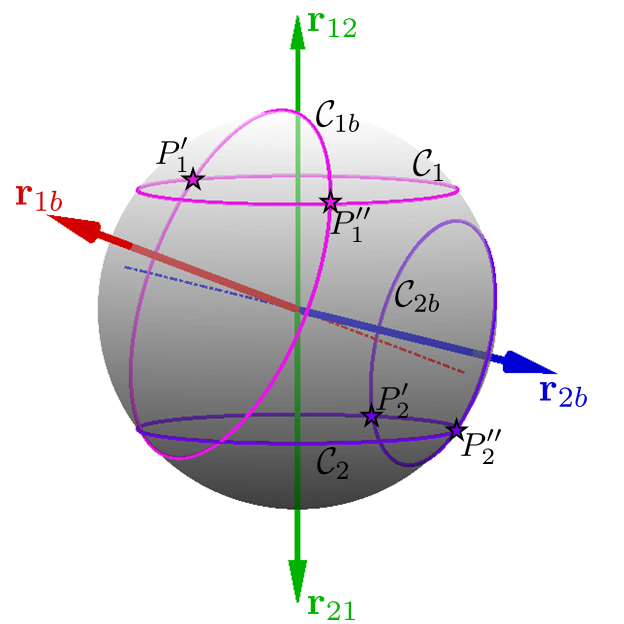}
  \caption{\small{Illustration of the constraints on effective shape space Variables. The \emph{red}, \emph{blue} and \emph{green} arrows represent directions of the relative position vectors $\mathbf{r}_{1b}$, $\mathbf{r}_{2b}$ and $\mathbf{r}_{12}$ (or $\mathbf{r}_{21}$), respectively.}}
  \vspace{-1.5em}
  \label{Fig:Constraints}
\end{center}
\end{figure}
%
%

%
%
\subsection{Closed Loop Dynamics on the Effective Shape Space}
%
Before going into detailed analysis of the dynamics at hand, we introduce the following simplifying assumptions\footnote{\editKG{We introduce these simplifying assumptions for the sake of mathematical tractability in this initial analysis. For future work, we intend to relax some of these assumptions to explore the broader space of possible system behaviors.}}:
\begin{itemize}
\item[(A1)] The controller gains ($\mu_i$ and $\mu_i^b$) are equal and common for both agents, i.e. $\mu_1 = \mu_2 = \mu_1^b = \mu_2^b = \mu$.
\item[(A2)] The bearing offset parameters with respect to the beacon are common for both agents, i.e. $a_{1b} = a_{2b} = a_0$.
\item[(A3)] The bearing offset parameters with respect to the other agent are the same for both agents, i.e. $a_1 = a_2 = a$.
\end{itemize}
Under these three assumptions (A1)-(A3), the following set of self-contained equations describe the closed-loop shape dynamics \editKG{on the effective shape space}:
\begin{subequations}
\begin{align}
\text{\footnotesize{ $
\dot{\rho} 
$}}
&
\text{\footnotesize{ $
= \bar{x}_1 + \bar{x}_{2} 
$}}
\label{eqn:2AgentShapeDynamics_rho}
\\
\text{\footnotesize{ $
\dot{\rho}_{1b} 
$}}
&
\text{\footnotesize{ $
= \bar{x}_{1b}
$}}
\label{eqn:2AgentShapeDynamics_rho-1b}
\\
\text{\footnotesize{ $
\dot{\rho}_{2b} 
$}}
&
\text{\footnotesize{ $
= \bar{x}_{2b}
$}}
\label{eqn:2AgentShapeDynamics_rho-2b}
\\
\text{\footnotesize{ $
\dot{\bar{x}}_1
$}}
&
\text{\footnotesize{ $
= \frac{\lambda}{\rho} \Bigl(1 - \tilde{x} - \bar{x}_1^2 - \bar{x}_1 \bar{x}_{2} \Bigr) - (1-\lambda)\mu (\bar{x}_1 - a) \left(1 - \bar{x}_1^2 \right)
$}}
\nonumber \\
& \qquad 
\text{\footnotesize{ $
-\lambda\mu(\bar{x}_{1b}-a_0)\left(\frac{\rho_{1b}^2 +\rho^2 - \rho_{2b}^2 }{2\rho \rho_{1b}} - \bar{x}_{1b}\bar{x}_1 \right), 
$}}
\label{eqn:2AgentShapeDynamics_x-1}
\\
\text{\footnotesize{ $
\dot{\bar{x}}_2
$}}
&
\text{\footnotesize{ $
= \frac{\lambda}{\rho} \Bigl(1 - \tilde{x} - \bar{x}_2^2 - \bar{x}_1 \bar{x}_{2} \Bigr) - (1-\lambda)\mu (\bar{x}_2 - a) \left(1 - \bar{x}_2^2 \right)  
$}}
\nonumber \\
& \qquad 
\text{\footnotesize{ $
-\lambda\mu(\bar{x}_{2b}-a_0)\left(\frac{\rho_{2b}^2 +\rho^2 - \rho_{1b}^2 }{2\rho \rho_{2b}} - \bar{x}_{2b}\bar{x}_2 \right),
$}}
\label{eqn:2AgentShapeDynamics_x-2}
\\
\text{\footnotesize{ $
\dot{\bar{x}}_{1b} 
$}}
&
\text{\footnotesize{ $
= - (1-\lambda) \Bigl(\mu (\bar{x}_1 - a) + \frac{1 - \tilde{x}}{\rho}\Bigr) \left(\frac{\rho_{1b}^2 +\rho^2 - \rho_{2b}^2 }{2\rho \rho_{1b}} - \bar{x}_{1b}\bar{x}_1 \right)
$}}
\nonumber \\
& \qquad 
\text{\footnotesize{ $
- (1-\lambda) \frac{\bar{x}_1}{\rho} \Biggl(\left(\frac{\rho_{2b}}{\rho_{1b}}\right)\bar{x}_{2b} - \left(\frac{\rho}{\rho_{1b}}\right)\bar{x}_{2}  - \bar{x}_{1b}\tilde{x} \Biggr) 
$}}
\nonumber \\
& \qquad 
\text{\footnotesize{ $
-\left(\lambda\mu(\bar{x}_{1b}-a_0) - \frac{1}{\rho_{1b}}\right)\Bigl(1 - \bar{x}_{1b}^2 \Bigr),
$}}
\label{eqn:2AgentShapeDynamics_x-1b}
\\
\text{\footnotesize{ $
\dot{\bar{x}}_{2b} 
$}}
& 
\text{\footnotesize{ $
= - (1-\lambda)\Bigl(\mu (\bar{x}_2 - a) + \frac{1 - \tilde{x}}{\rho}\Bigr) \left(\frac{\rho_{2b}^2 +\rho^2 - \rho_{1b}^2 }{2\rho \rho_{2b}} - \bar{x}_{2b}\bar{x}_2 \right)
$}}
\nonumber \\
& \qquad 
\text{\footnotesize{ $
- (1-\lambda)\frac{\bar{x}_2}{\rho} \Biggl(\left(\frac{\rho_{1b}}{\rho_{2b}}\right)\bar{x}_{1b} - \left(\frac{\rho}{\rho_{2b}}\right)\bar{x}_{1}  - \bar{x}_{2b}\tilde{x} \Biggr) 
$}}
\nonumber \\
& \qquad 
\text{\footnotesize{ $
-\left(\lambda\mu(\bar{x}_{2b}-a_0) - \frac{1}{\rho_{2b}}\right)\Bigl(1 - \bar{x}_{2b}^2 \Bigr),
$}}
\label{eqn:2AgentShapeDynamics_x-2b}
\\
\text{\footnotesize{ $
\dot{\tilde{x}}
$}}
&
\text{\footnotesize{ $
= -\lambda\mu(\bar{x}_{2b}-a_0)\left(-\left(\frac{\rho}{\rho_{2b}}\right)\bar{x}_{1} + \left(\frac{\rho_{1b}}{\rho_{2b}}\right)\bar{x}_{1b}  - \bar{x}_{2b}\tilde{x}\right)
$}}
\nonumber \\
& \qquad 
\text{\footnotesize{ $
-\lambda\mu(\bar{x}_{1b}-a_0)\left(\left(\frac{\rho_{2b}}{\rho_{1b}}\right)\bar{x}_{2b} + \left(\frac{\rho}{\rho_{1b}}\right)(-\bar{x}_{2})  - \bar{x}_{1b}\tilde{x}\right) 
$}}
\nonumber \\
& \qquad 
\text{\footnotesize{ $
- (1-\lambda) \Biggl[ \left(\mu (\bar{x}_1 - a) + \frac{1 - \tilde{x}}{\rho} \right) \left(-\bar{x}_2 - \tilde{x}\bar{x}_1 \right) \Biggr.
$}}
\nonumber \\
& \qquad \;
\text{\footnotesize{ $
\Biggl. + \bar{x}_1 \left(\frac{1-\tilde{x}^2}{\rho}\right) \Biggr]
- (1-\lambda) \Biggl[ \bar{x}_2 \left(\frac{1-\tilde{x}^2}{\rho}\right) \Biggr.
$}}
\nonumber \\
& \qquad \;
\text{\footnotesize{ $
\Biggl. + \left(\mu (\bar{x}_2 - a) + \frac{1 - \tilde{x}}{\rho} \right) \left(-\bar{x}_1 - \tilde{x}\bar{x}_2 \right) \Biggr]
$}}.
\label{eqn:2AgentShapeDynamics_x-tilde}
\end{align}
\end{subequations}
%
%
%

%
%
%
\section{Existence of Circling Equilibria}
\label{sec:Rel_Equilibrium}
%
The rest of this work is focused on determining conditions for existence of equilibria for the closed-loop dynamics \eqref{eqn:2AgentShapeDynamics_rho}-\eqref{eqn:2AgentShapeDynamics_x-tilde}. These equilibria correspond to the agents moving on circular orbits with a common radius, in planes perpendicular to a common axis passing through the beacon, and therefore we will refer to them as \textit{circling equilibria}. We proceed by setting $\dot{\rho} = \dot{\rho}_{1b} = \dot{\rho}_{2b} = 0$, which yields 
\begin{align}
\bar{x}_2 = -\bar{x}_1, \quad \bar{x}_{1b} = 0 = \bar{x}_{2b}.
\label{EquiliBr_x-ib_N_x-i}
\end{align}
\editKG{If ${\bf r}_1$, ${\bf r}_2$, and ${\bf r}_b$ are collinear, then it is clear that the equilibrium constraint $\bar{x}_{1b} = 0 = \bar{x}_{2b}$ implies that $\bar{x}_{1} = 0 = \bar{x}_{2}$. If the agents and beacon are not collinear, then the equilibrium constraint $\bar{x}_{1b} = 0 = \bar{x}_{2b}$ implies that} the circles $\mathcal{C}_{1b}$ and $\mathcal{C}_{2b}$ will be two great circles on the unit sphere, which intersect at two distinct antipodal points (see Figure~\ref{Fig:Constraints}). Moreover, \editKG{the equilibrium constraint $\bar{x}_2 = -\bar{x}_1$ implies that} the circles $\mathcal{C}_1$ and $\mathcal{C}_2$ coincide at every relative equilibrium. Substituting \eqref{EquiliBr_x-ib_N_x-i} into \eqref{eqn:2AgentShapeDynamics_x-1}-\eqref{eqn:2AgentShapeDynamics_x-tilde} and simplifying, the closed loop dynamics on the nullclines $\dot{\rho} = \dot{\rho}_{1b} = \dot{\rho}_{2b} = 0$ can be expressed as
\begin{align}
\label{eqn:2AgentDynamicsAtEquilibriumCommonCBParameter}
\dot{\bar{x}}_1
&= -(1-\lambda)\mu (\bar{x}_1 - a) \left(1 - \bar{x}_1^2 \right)  \nonumber \\
& \qquad + \lambda\mu a_0\left(\frac{\rho_{1b}^2 +\rho^2 - \rho_{2b}^2 }{2\rho \rho_{1b}} \right) + \frac{\lambda}{\rho} \Bigl(1 - \tilde{x}\Bigr),  \nonumber \\
\dot{\bar{x}}_2
&= -(1-\lambda)\mu (-\bar{x}_1 - a) \left(1 - \bar{x}_1^2 \right)  \nonumber \\
& \qquad+ \lambda\mu a_0\left(\frac{-\rho_{1b}^2 +\rho^2 + \rho_{2b}^2 }{2\rho \rho_{2b}} \right) + \frac{\lambda}{\rho} \Bigl(1 - \tilde{x}\Bigr),  \nonumber \\
\dot{\bar{x}}_{1b} 
&= -(1-\lambda)\Bigl(\mu (\bar{x}_1 - a) + \frac{1}{\rho}\left(1 - \tilde{x}	\right)\Bigr)\left(\frac{\rho_{1b}^2 +\rho^2 - \rho_{2b}^2 }{2\rho \rho_{1b}} \right) \nonumber \\
& \qquad - \frac{(1-\lambda)\bar{x}_1^2}{\rho_{1b}} + \lambda\mu a_0 + \frac{1}{\rho_{1b}}, \nonumber \\
\dot{\bar{x}}_{2b} 
&= -(1-\lambda)\Bigl(\mu (-\bar{x}_1 - a) + \frac{1}{\rho}\left(1 - \tilde{x}	\right)\Bigr)\left(\frac{\rho_{2b}^2 + \rho^2 - \rho_{1b}^2 }{2\rho \rho_{2b}} \right) \nonumber \\
& \qquad- \frac{(1-\lambda)\bar{x}_1^2}{\rho_{2b}} + \lambda\mu a_0 + \frac{1}{\rho_{2b}}, \nonumber \\
\dot{\tilde{x}}
&= -\mu \bar{x}_1\Bigl(2(1-\lambda)\bar{x}_1   + \frac{\lambda a_0 \rho }{\rho_{1b}\rho_{2b}} \left(\rho_{2b}-\rho_{1b} \right)\Bigr).  
\end{align}
Then by narrowing our focus to the special case when $a_0 = 0$, we arrive at the following result.
\begin{proposition}
Consider a beacon-referenced mutual CB pursuit system with shape dynamics \eqref{eqn:2AgentShapeDynamics_rho}-\eqref{eqn:2AgentShapeDynamics_x-tilde} parametrized by $\mu$, $\lambda$, and the CB \editKG{parameters $a$, $a_0$, with $a_0=0$}. Then, a \textit{circling equilibrium} exists if and only if $a<0$, and the corresponding equilibrium values \editKG{satisfy}
\begin{equation}
\begin{aligned}
\bullet \quad &
\bar{x}_1 = \bar{x}_2 = 0,
\quad 
\bar{x}_{1b} = \bar{x}_{2b} = 0,
\quad 
\tilde{x} = -1,
\\
\bullet \quad &
\rho_{1b} = \rho_{2b},
\quad
\rho = \frac{2\lambda }{(1-\lambda)\mu (-a)}.
\end{aligned}
\label{eqn:a0SameAlphaCircValues}
\end{equation}
\label{prop:existenceProp_1}
\end{proposition}
\begin{proof}
In this case, it is clear that $\dot{\tilde{x}}=0$ if and only if $\bar{x}_1=0$, and therefore, from \eqref{eqn:2AgentDynamicsAtEquilibriumCommonCBParameter}, we can conclude that the following conditions must hold true at an equilibrium
\begin{subequations}
\begin{align}
(1-\lambda)\mu a  + \frac{\lambda}{\rho} \Bigl(1 - \tilde{x}\Bigr)
&=
0,
\label{eqn:2AgentDynamicsAtEquilibriumCommonCBParameter_a0Zero__1}
\\
(1-\lambda)\left( \frac{1 - \tilde{x}}{\rho} - \mu a \right)\left(\frac{\rho_{1b}^2 +\rho^2 - \rho_{2b}^2 }{2\rho \rho_{1b}} \right)
&=
\frac{1}{\rho_{1b}}, 
\label{eqn:2AgentDynamicsAtEquilibriumCommonCBParameter_a0Zero__2}
\\
(1-\lambda)\left( \frac{1 - \tilde{x}}{\rho} -\mu a \right)\left(\frac{\rho_{2b}^2 +\rho^2 - \rho_{1b}^2 }{2\rho \rho_{2b}} \right)
&=
\frac{1}{\rho_{2b}}.
\label{eqn:2AgentDynamicsAtEquilibriumCommonCBParameter_a0Zero__3}
\end{align}
\end{subequations}
Now, if $\tilde{x} = 1$, the first condition \eqref{eqn:2AgentDynamicsAtEquilibriumCommonCBParameter_a0Zero__1} holds true if and only if $a=0$. But with these choices for $\tilde{x}$ and $a$, the last two conditions \eqref{eqn:2AgentDynamicsAtEquilibriumCommonCBParameter_a0Zero__2}-\eqref{eqn:2AgentDynamicsAtEquilibriumCommonCBParameter_a0Zero__3} lead to 
$\frac{1}{\rho_{1b}} = \frac{1}{\rho_{2b}} = 0$, which cannot be true since both $\rho_{1b}$ and $\rho_{2b}$ are finite. Therefore we must have $\tilde{x} \neq 1$ at an equilibrium, and then the first condition \eqref{eqn:2AgentDynamicsAtEquilibriumCommonCBParameter_a0Zero__1} yields the equilibrium value of $\rho$ as
\begin{equation}
\rho = \frac{\lambda(1 - \tilde{x})}{(1-\lambda)\mu (-a)}.
\label{eqn:rhoEquilibrium_a0zero}
\end{equation}
As $\rho$ must be positive and finite, \eqref{eqn:rhoEquilibrium_a0zero} yields a meaningful solution if and only if $a < 0$. Substituting this solution for $\rho$ into \eqref{eqn:2AgentDynamicsAtEquilibriumCommonCBParameter_a0Zero__2}-\eqref{eqn:2AgentDynamicsAtEquilibriumCommonCBParameter_a0Zero__3}, we have 
\begin{subequations}
\begin{align}
\frac{1}{2 \rho_{1b}}\biggl[-\left(1 - \tilde{x}	\right)\left(\frac{\rho_{1b}^2 - \rho_{2b}^2 }{\rho^2} \right) + 1 + \tilde{x} \biggr]
&=0,
\label{eqn:EQ_Final_Conditions_on_Nullcline__1}
\\
\frac{1}{2 \rho_{2b}}\biggl[-\left(1 - \tilde{x}	\right)\left(1 - \frac{\rho_{1b}^2 - \rho_{2b}^2 }{2\rho^2} \right) + 2 \biggr]
&= 0.  
\label{eqn:EQ_Final_Conditions_on_Nullcline__2}
\end{align}
\end{subequations}
Clearly \eqref{eqn:EQ_Final_Conditions_on_Nullcline__1} holds true if and only if either of the following conditions hold:
\\
\indent (I) $\displaystyle \rho_{1b} = \rho_{2b}$ and $\displaystyle \tilde{x} = -1$, or
\\
\indent (II) $\displaystyle \rho_{1b} \neq \rho_{2b}$ and $\displaystyle \left(\frac{\rho_{1b}^2 - \rho_{2b}^2 }{\rho^2} \right) = \frac{1+\tilde{x}}{1-\tilde{x}}$, \editKG{with $\tilde{x} \neq -1$}.

Then it is straightforward to verify that the first set of conditions (I) satisfy \eqref{eqn:EQ_Final_Conditions_on_Nullcline__2}. However, by substituting the second set of conditions (II) into \eqref{eqn:EQ_Final_Conditions_on_Nullcline__2}, we have
\begin{align}
&&
-(1-\tilde{x}) \left(1 - \frac{1}{2}\left(\frac{1+\tilde{x}}{1-\tilde{x}}\right) \right) +2 
&= 0
\nonumber \\
\Rightarrow &&
-(1-\tilde{x}) + \frac{1+\tilde{x}}{2}  + 2 
&= 0
\nonumber \\
\Rightarrow &&
\frac{3}{2}\left( \tilde{x} + 1 \right)
&= 0,
\end{align}
which is true if and only if $\tilde{x} = -1$. But, this contradicts the stated condition (II). Therefore this option is not viable, and (I) must hold true at an equilibrium.

\editKG{Lastly, it is clear that the proposed equilibrium values \eqref{eqn:a0SameAlphaCircValues} satisfy the constraints \eqref{Constraints_2a}-\eqref{Constraints_2b} and are therefore valid solutions.} This concludes our proof.
\end{proof}
\vspace{.25cm}
\begin{figure}[t!]
\begin{center}
  \includegraphics[width=0.45\textwidth]{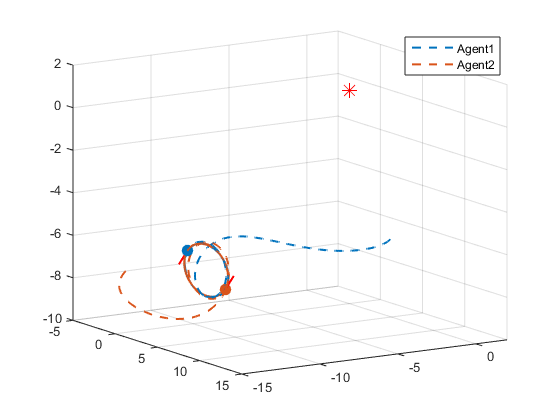}
  \caption{\small{Illustration of the type of circling equilibrium described in Proposition \ref{prop:existenceProp_1}. The asterisk denotes the location of the beacon, and the agents follow a circling trajectory on a plane perpendicular to the axis which passes through the beacon.}}
  \vspace{-1.5em}
  \label{Fig:OffsetCircle}
\end{center}
\end{figure}
\editKG{Figure \ref{Fig:OffsetCircle} illustrates the type of circling equilibrium which is described in Proposition \ref{prop:existenceProp_1}. Note that the values for $\rho_{1b}$ and $\rho_{2b}$ (i.e. the distance of each agent from the beacon) are the same, but the particular values are determined by initial conditions. However, the separation between the agents (i.e. $\rho$) is determined by the control parameters.}

\vspace{.25cm}

We now shift our attention to the case where $a_0 \neq 0$, and show that circling equilibria exist in this scenario as well.
\begin{proposition}
Consider a beacon-referenced mutual CB pursuit system with shape dynamics \eqref{eqn:2AgentShapeDynamics_rho}-\eqref{eqn:2AgentShapeDynamics_x-tilde} parametrized by $\mu$, $\lambda$, and the CB \editKG{parameters $a$, $a_0$, with $a_0\neq0$}. The following statements are true.
\\
\indent (a) Whenever $(1-\lambda) a + \lambda a_0 < 0$, a circling equilibrium exists, and the corresponding equilibrium values are given by
\begin{equation}
\begin{aligned}
\bullet \quad &
\bar{x}_1 = \bar{x}_2 = 0,
\quad 
\bar{x}_{1b} = \bar{x}_{2b} = 0,
\quad 
\tilde{x} = -1,
\\
\bullet \quad &
\rho_{1b} = \rho_{2b} 
= \frac{\lambda}{-\mu \bigl((1-\lambda)a + \lambda a_0\bigr)},
\\
\bullet \quad &
\rho 
= 2\rho_{1b} = \frac{2\lambda}{-\mu \bigl((1-\lambda)a + \lambda a_0\bigr)}.
\end{aligned}
\end{equation}
\\
\indent (b) Whenever $a_0 < 0$, $a>0$, and $(1-\lambda) a + \lambda a_0 < 0$, a circling equilibrium exists, and the corresponding equilibrium values are given by
\begin{equation}
\begin{aligned}
\bullet \quad &
\bar{x}_1 = \bar{x}_2 = 0,
\quad 
\bar{x}_{1b} = \bar{x}_{2b} = 0,
\quad 
\tilde{x} = 1,
\\
\bullet \quad &
\rho_{1b} = \rho_{2b}
= \frac{\lambda a_0}{\mu\Bigl((1-\lambda)^2 a^2 - \lambda^2 a_0^2\Bigr)},
\\
\bullet \quad &
\rho 
= \frac{-2(1-\lambda) a}{\mu\Bigl((1-\lambda)^2 a^2 - \lambda^2 a_0^2\Bigr)}.
\end{aligned}
\end{equation}
\label{prop:existenceProp_2}
\end{proposition}
\begin{proof}
It directly follows from \eqref{eqn:2AgentDynamicsAtEquilibriumCommonCBParameter} that $\dot{\tilde{x}}=0$ if $\bar{x}_1 = 0$, and in that situation we can express the closed loop dynamics on the nullclines $\dot{\rho} = \dot{\rho}_{1b} = \dot{\rho}_{2b} = \dot{\tilde{x}}=0$ as
\begin{align}
\dot{\bar{x}}_1
&= (1-\lambda)\mu a + \lambda\mu a_0\left(\frac{\rho_{1b}^2 +\rho^2 - \rho_{2b}^2 }{2\rho \rho_{1b}} \right) + \frac{\lambda}{\rho} \Bigl(1 - \tilde{x}\Bigr),  \nonumber \\
\dot{\bar{x}}_2
&= (1-\lambda)\mu a + \lambda\mu a_0\left(\frac{\rho_{2b}^2 +\rho^2 - \rho_{1b}^2 }{2\rho \rho_{2b}} \right) + \frac{\lambda}{\rho} \Bigl(1 - \tilde{x}\Bigr),  \nonumber \\
\dot{\bar{x}}_{1b} 
&= 
(1-\lambda)\left( \mu a - \frac{1 - \tilde{x}}{\rho} \right) \left(\frac{\rho_{1b}^2 + \rho^2 - \rho_{2b}^2}{2\rho \rho_{1b}} \right)  
\nonumber \\
& \qquad + \lambda\mu a_0 + \frac{1}{\rho_{1b}}, \nonumber \\
\dot{\bar{x}}_{2b} 
&= 
(1-\lambda)\left( \mu a - \frac{1 - \tilde{x}}{\rho} \right) \left(\frac{\rho_{2b}^2 + \rho^2 - \rho_{1b}^2}{2\rho \rho_{2b}} \right)  
\nonumber \\
& \qquad + \lambda\mu a_0 + \frac{1}{\rho_{2b}}.   
\label{eqn:2AgentDynamicsAtEquilibriumCommonCBParameter_x1_zero}
\end{align}

We note that taking the difference of $\dot{\bar{x}}_1-\dot{\bar{x}}_2$ yields
\begin{align}
\label{eqn:x1DotDiff}
\dot{\bar{x}}_1-\dot{\bar{x}}_2 
&= 
\lambda\mu a_0 \left( \frac{\rho_{1b}^2 +\rho^2 - \rho_{2b}^2 }{2\rho \rho_{1b}} - \frac{\rho_{2b}^2 +\rho^2 - \rho_{1b}^2 }{2\rho \rho_{2b}} \right) 
\nonumber \\
%
%
&= 
\frac{\lambda\mu a_0(\rho_{1b}-\rho_{2b}) }{2\rho\rho_{1b}\rho_{2b}} \Bigl(\left(\rho_{1b} + \rho_{2b} \right)^2 -\rho^2 \Bigr),
\end{align}
and similar calculations lead to 
\begin{align}
\label{eqn:x1bDotDiff}
\dot{\bar{x}}_{1b}-\dot{\bar{x}}_{2b} 
&= 
(1-\lambda) \left( \mu a - \frac{1 - \tilde{x}}{\rho} \right) \left(\frac{\rho_{1b}-\rho_{2b}}{2\rho\rho_{1b}\rho_{2b}}\right) 
\nonumber \\
&\qquad 
\times \Bigl( (\rho_{1b} + \rho_{2b})^2 -\rho^2 \Bigr) + \frac{\rho_{2b}-\rho_{1b}}{\rho_{1b}\rho_{2b}}.
\end{align}
Then by setting both \eqref{eqn:x1DotDiff} and \eqref{eqn:x1bDotDiff} equal to zero, i.e. by setting the derivatives of $\bar{x}_2$ and $\bar{x}_{2b}$ identical to the derivatives of $\bar{x}_1$ and $\bar{x}_{1b}$ respectively, we can conclude that $\rho_{2b}$ must be equal to $\rho_{1b}$ at an equilibrium. Substituting this equivalence into \eqref{eqn:2AgentDynamicsAtEquilibriumCommonCBParameter_x1_zero}, we can further conclude that the following conditions must hold true at an equilibrium
\begin{align}
(1-\lambda)\mu a + \lambda\mu a_0 \left( \frac{\rho}{2\rho_{1b}} \right) + \lambda \left(\frac{1 - \tilde{x}}{\rho} \right)
&= 0,
\label{eqn:equilibriumDynamicsTwoAgentA0NonZeroTerm_1}
\\
(1-\lambda) \left(\mu a - \frac{1 - \tilde{x}}{\rho} \right) \left(\frac{\rho}{2 \rho_{1b}} \right)  + \lambda\mu a_0 + \frac{1}{\rho_{1b}}
&= 0.
\label{eqn:equilibriumDynamicsTwoAgentA0NonZeroTerm_2}
\end{align}

\editKG{If the two agents and the beacon are collinear, then the constraint $\rho_{1b} = \rho_{2b}$ implies that $\rho = 2\rho_{1b}$. Substituting this equivalence into 
\eqref{eqn:equilibriumDynamicsTwoAgentA0NonZeroTerm_1} and \eqref{eqn:equilibriumDynamicsTwoAgentA0NonZeroTerm_2} yields
\begin{align}
(1-\lambda)\mu a + \lambda\mu a_0  + \lambda \left(\frac{1 - \tilde{x}}{2\rho_{1b}} \right)
&= 0,
\label{eqn:equilibriumDynamicsTwoAgentA0NonZeroTerm_1_subs}
\\
(1-\lambda)\mu a + \lambda\mu a_0  - (1-\lambda)\left(\frac{1 - \tilde{x}}{2\rho_{1b}}  \right) + \frac{1}{\rho_{1b}}
&= 0,
\label{eqn:equilibriumDynamicsTwoAgentA0NonZeroTerm_2_subs}
\end{align}
from which it follows that $\tilde{x}=-1$. Substituting this value back into \eqref{eqn:equilibriumDynamicsTwoAgentA0NonZeroTerm_1_subs} results in
\begin{align}
\rho_{1b} &= \frac{\lambda}{-\mu \bigl((1-\lambda)a + \lambda a_0\bigr)},
\end{align}
and this is a meaningful solution if and only if $(1-\lambda)a + \lambda a_0 < 0$. Since these values satisfy all constraints introduced at the beginning of Section \ref{sec:Constraints}, part (a) of the proposition is established.}

\editKG{If the agents and beacon are not collinear, then the equilibrium constraint $\bar{x}_1=\bar{x}_2=0$ implies that the circles $\mathcal{C}_{1}$ and $\mathcal{C}_{2}$ will coincide \editKG{as} a great circle around $\mathbf{r}_{12}$ (see Figure~\ref{Fig:Constraints}). Moreover, the axes of $\mathcal{C}_{1b}$, $\mathcal{C}_{2b}$ and $\mathcal{C}_{1}$ lie on the same plane, which enforces that these three great circles will intersect each other at two antipodal points. As a consequence, $\tilde{x}$ can be either $1$ or $-1$ at such an equilibrium.}

If $\tilde{x}=1$, then \eqref{eqn:equilibriumDynamicsTwoAgentA0NonZeroTerm_1} allows us to express $\rho$ as
\begin{equation}
\rho 
= 
-2 \left(\frac{1-\lambda}{\lambda} \right)\left(\frac{a}{a_0} \right)\rho_{1b}.
\label{Existence_4.2_Cond1}
\end{equation}
As both $\rho$ and $\rho_{1b}$ must be positive, \eqref{Existence_4.2_Cond1} is meaningful if and only if $a/a_0 < 0$. Also, since $\rho_{1b} = \rho_{2b}$, substituting \eqref{Existence_4.2_Cond1} into constraint \eqref{Constraints_2c} yields 
\begin{align}
\left(\frac{1-\lambda}{\lambda}\right) \left(\frac{-a}{a_0}\right) < 1,
\label{eqn:prop42partBconstraint}
\end{align}
with the strict inequality resulting from the fact that we have assumed that the agents and beacon are not collinear. The combination of \eqref{eqn:prop42partBconstraint} with $a/a_0 < 0$ yields two possibilities:
\begin{itemize}
\item Case 1: $a_0>0$, $a<0$, $(1-\lambda) a + \lambda a_0 > 0$;
\item Case 2: $a_0<0$, $a>0$, $(1-\lambda) a + \lambda a_0 < 0$.
\end{itemize}
Also, substitution of \eqref{Existence_4.2_Cond1} into \eqref{eqn:equilibriumDynamicsTwoAgentA0NonZeroTerm_2} leads to
\begin{equation}
-\mu\left(\frac{(1-\lambda)^2}{\lambda} \right) \left(\frac{a^2}{a_0} \right) + \lambda\mu a_0 + \frac{1}{\rho_{1b}}
=
0,
\end{equation}
which in turn yields
\begin{align}
\rho_{1b} 
= \frac{1}{\mu\left(\frac{(1-\lambda)^2}{\lambda} \right) \left(\frac{a^2}{a_0} \right) - \lambda\mu a_0} \nonumber \\
= \frac{\lambda a_0}{\mu\Bigl((1-\lambda)^2 a^2 - \lambda^2 a_0^2\Bigr)}.
\end{align}
This yields a meaningful solution if and only if 
\begin{equation}
a_0\Bigl((1-\lambda)^2 a^2 - \lambda^2 a_0^2\Bigr) > 0, 
\end{equation}
i.e. 
\begin{equation}
a_0\Bigl((1-\lambda) a - \lambda a_0\Bigr)\Bigl((1-\lambda) a + \lambda a_0\Bigr) > 0. 
\end{equation}
It is straightforward to verify that Case 2 (but not Case 1) satisfies this constraint, leading to the conditions of part (b) of the proposition.

On the other hand, if $\tilde{x}=-1$, then \eqref{eqn:equilibriumDynamicsTwoAgentA0NonZeroTerm_1}-\eqref{eqn:equilibriumDynamicsTwoAgentA0NonZeroTerm_2} simplifies to
\begin{align}
\frac{1}{2\rho\rho_{1b}}\Bigl[2\rho_{1b}\bigl((1-\lambda)\mu a \rho  + 2\lambda\bigr) + \lambda \mu a_0 \rho^2\Bigr]
&=
0,
\label{Existence_4.2_Cond2}
\\
\frac{1}{2\rho_{1b}}\Bigl[(1-\lambda)\mu a \rho  + 2\lambda + 2\lambda\mu a_0 \rho_{1b} \Bigr]
&=
0.
\label{Existence_4.2_Cond3}
\end{align}
\editKG{However, it follows from \eqref{Existence_4.2_Cond2}-\eqref{Existence_4.2_Cond3} that at an equilibrium we must have $\rho = 2\rho_{1b}$, i.e. this corresponds to the collinear configuration addressed earlier in part (a) of the proposition. (Note that the condition $\tilde{x}=\pm 1$ is a necessary condition of the agents and beacon being a non-collinear configuration, but it is not sufficient.)} This completes the proof.
\end{proof}
\vspace{.25cm}
\begin{figure}[t!]
\begin{center}
  \includegraphics[width=0.45\textwidth]{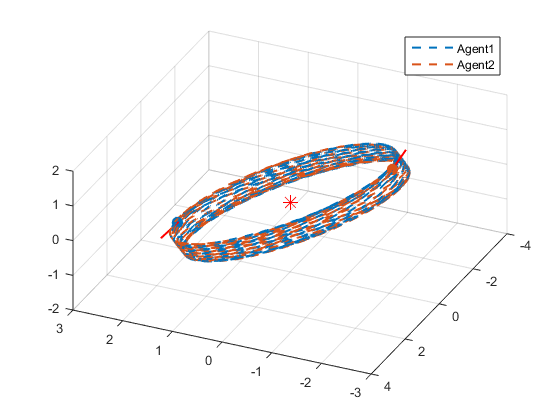}
  \caption{\small{Illustration of the type of circling equilibrium described in Proposition \ref{prop:existenceProp_2}, part (a). The asterisk denotes the location of the beacon, and the agents slowly converge to a planar circling trajectory centered on the beacon.}}
  \vspace{-1.5em}
  \label{Fig:CenteredCircle}
\end{center}
\end{figure}
\begin{figure}[t!]
\begin{center}
  \includegraphics[width=0.45\textwidth]{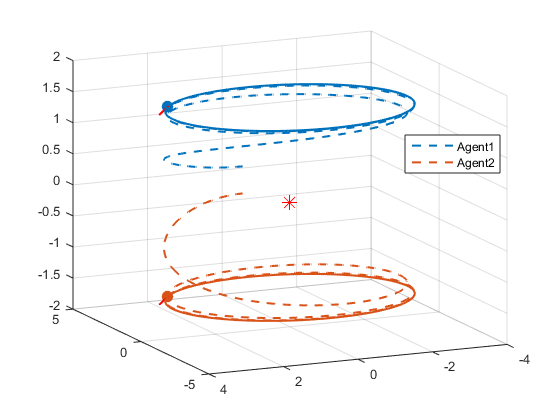}
  \caption{\small{Illustration of the type of circling equilibrium described in Proposition \ref{prop:existenceProp_2}, part (b). Simulations indicate that this equilibrium has a very small region of attraction.}}
  \vspace{-1.5em}
  \label{Fig:StackedCircle}
\end{center}
\end{figure}
\editKG{Figure \ref{Fig:CenteredCircle} and \ref{Fig:StackedCircle}	depict MATLAB simulations corresponding to the equilibria described in Proposition \ref{prop:existenceProp_2} part (a) and part (b), respectively. Simulations indicate that the type of circling equilibria described in part (a) of the proposition (i.e. Fig.\ref{Fig:CenteredCircle}) are attractive, but that convergence occurs on a long time-scale, while the equilibria described in part (b) of the proposition (i.e. Fig. \ref{Fig:StackedCircle}) have a very small region of attraction. A rigorous stability analysis will be carried out in future work.}
\begin{remark}
We note that Proposition \ref{prop:existenceProp_2} provides only sufficient (and not necessary) conditions for existence of circling equilibria in the case $a_0\neq 0$. This stems from the fact that there remains another possibility for $\dot{\tilde{x}}=0$ in \eqref{eqn:2AgentDynamicsAtEquilibriumCommonCBParameter}, namely 
$\displaystyle \bar{x}_1 \neq 0$, $\displaystyle \rho_{1b} \neq \rho_{2b}$, and $\displaystyle \frac{\lambda a_0 \rho }{\rho_{1b}\rho_{2b}} \left(\rho_{2b}-\rho_{1b} \right) + 2(1-\lambda)\bar{x}_1 = 0$. In future work we will analyze this case to determine whether it presents a legitimate additional solution corresponding to circling equilibria. 
\end{remark}
%
%
%

%
%
%
\section{Conclusion}
\label{sec:Conclusion}
In this paper we have presented a new control law which implements a beacon-referenced version of the CB control law. We have analyzed the 2-agent (with beacon) system and demonstrated the existence of circling equilibria for particular parameter choices. The circling equilibria obtained in this setting have a radius determined by control parameters rather than by initial conditions (as was the case for mutual CB pursuit without a beacon \cite{CP_3D_Kevin_2010}), which offers a method for designing circling trajectories with a desired diameter.

Future work will focus on stability analysis for the special solutions presented here. Additional directions for research include exploration of the solution space for systems with $a_1 \neq a_2$, as well as analysis of the beacon-referenced cyclic pursuit system (i.e. $n>2$).

\bibliographystyle{IEEEtran}
\bibliography{Draft_Refs.bib}
\end{document}